\documentclass[11pt]{article}
\usepackage{natbib,amssymb,amsmath,amsthm,fullpage,hyperref,graphicx}

\newcommand{\Cov}[1]{{\text{Cov}}[ \ensuremath{ #1 } ]  }

\newtheorem{lemma}{Lemma}

\title{
Lasso, fractional norm and structured sparse estimation 
using
a 
Hadamard product parametrization}
\author{Peter D. Hoff  \\
Department of Statistical Science \\
Duke University}
\date{\today}

\begin{document}

\maketitle

\begin{abstract}
Using a multiplicative reparametrization, I show that 
a subclass of $L_q$ penalties with $q\leq 1$ can be expressed
as sums of $L_2$ penalties. 
It follows that the lasso 
and other norm-penalized regression estimates 
may be obtained using a very simple and intuitive 
alternating ridge regression algorithm. 
As compared to a similarly intuitive EM algorithm 
for $L_q$ optimization, the proposed algorithm 
avoids some numerical instability issues and is 
also competitive in terms of speed. 
Furthermore, the proposed  algorithm  
can be extended to accommodate sparse high-dimensional scenarios, 
generalized linear models, and can be used to create  
structured sparsity via penalties derived from 
covariance models for the parameters. 
Such model-based penalties 
may be useful for sparse estimation of  spatially or temporally 
structured parameters.  

\smallskip

\end{abstract}

\section{Introduction}
Consider estimation for the normal linear regression model
$y \sim N_n( X \beta, \sigma^2 I)$, where 
$X \in \mathbb R^{n\times p}$ is a matrix 
of predictor variables and $\beta\in \mathbb R^p$ 
is a vector of regression coefficients to be estimated.  
A least squares estimate is a minimizer of the 
residual sum of squares $|| y- X\beta||^2.$
A popular alternative estimate is the lasso estimate
\citep{tibshirani_1996}, which
minimizes 
$|| y- X\beta||^2 + \lambda ||\beta||_1$,
a penalized residual sum of squares that 
balances fit to the data against the possibility that 
some or many of the elements of $\beta$ are small or zero. 
Indeed, minimizers of this penalized sum of squares  may have elements 
that are exactly zero.   

There exists a large variety of optimization algorithms for finding 
lasso estimates
(see \citet{schmidt_fung_rosales_2007} for a review). 
However, the details of many of these algorithms are 
somewhat opaque to data analysts who are not well-versed in the 
theory of optimization. 
One exception is the 
local
quadratic approximation (LQA) algorithm 
of \citet{fan_li_2001}, which proceeds by 
iteratively computing a series of ridge regressions. 
\citet{fan_li_2001} also suggested using 
LQA for non-convex $L_q$ penalization when $q<1$,  
and this technique was used by 
\citet{kaban_durrant_2008} and \citet{kaban_2013} 
in their studies of non-convex $L_q$-penalized 
logistic regression. 
However, 
LQA can be numerically unstable for some combinations of 
models and penalties. 
To remedy this, \citet{hunter_li_2005} 
suggested optimizing a surrogate ``perturbed''  
objective function. This perturbation 
must be user-specified, and its value can affect 
the parameter estimate. 
As an alternative to using local quadratic approximations, 
\citet{zou_li_2008} suggest $L_q$-penalized optimization using local 
linear approximations (LLA). While this approach avoids the instability 
of LQA, the algorithm is implemented by iteratively 
solving a series of $L_1$ penalization problems for which 
an optimization algorithm must be chosen as well. 

This article develops a simple alternative technique for 
obtaining $L_q$-penalized regression estimates for many values of 
$q\leq 1$. The technique is based on a 
non-identifiable 
Hadamard product 
parametrization (HPP) of $\beta$ as $\beta = u\circ v$, where 
``$\circ$'' denotes the Hadamard (element-wise) product 
of the vectors $u$ and $v$. As shown in the next section, 
if $\hat u$ and $\hat v$ are
optimal $L_2$-penalized  values of 
$u$ and $v$, 
then $\hat \beta = \hat u\circ \hat v$ 
is an optimal $L_1$-penalized value of $\beta$. 
An alternating ridge regression algorithm for 
obtaining $\hat u\circ \hat v$ is easy to understand and 
implement, and is competitive with LQA 
in terms of speed. Furthermore, a modified version of HPP 
can be adapted to 
provide fast convergence in sparse, high-dimensional scenarios. 
In Section 3 we consider extensions of this algorithm 
for non-convex $L_q$-penalized regression with $q\leq 1$. 
As in the $L_1$ case, 
$L_q$-penalized linear regression estimates may be found using 
alternating ridge regression, 
whereas estimates in generalized linear models can be obtained 
with a modified version of an iteratively reweighted least squares
algorithm. 
In Section 4 we show how the HPP
can facilitate structured sparsity in parameter estimates: 
The $L_2$ penalty on the vectors $u$ and $v$ can be interpreted 
as independent Gaussian prior distributions on the elements of 
$u$ and $v$. 
If instead we choose a penalty that mimics 
a dependent Gaussian prior, then we can achieve structured 
sparsity among the elements of $\hat \beta= \hat u\circ\hat v$. This 
technique is illustrated with an analysis of brain imaging data, 
for which a spatially structured HPP penalty 
is able to identify spatially contiguous regions of differential brain 
activity. A discussion 
follows in Section 5.

\section{$L_1$ optimization using the HPP and ridge regression}  
\subsection{The Hadamard product parametrization} 
The lasso  or $L_1$-penalized regression estimate $\hat \beta$ 
of $\beta$ for the 
model 
$y\sim N_p( X\beta,\sigma^2 I)$ is the minimizer 
of $|| y - X\beta||^2 + \lambda||\beta||_1$, 
or equivalently of the objective function 
\begin{equation}
f(\beta) =  \beta^\top Q \beta - 2 \beta^\top l  +  \lambda ||\beta||_1
\label{eqn:l1_obj}, 
\end{equation}
where $Q= X^\top X$ and $l=  X^\top y$. 
Now reparametrize the model so that 
$\beta= u \circ v$, where 
``$\circ$'' is the Hadamard (element-wise) product. 
We refer to this parametrization as the Hadamard product parametrization
(HPP).
Estimation of 
$u$ and $v$ using $L_2$ penalties corresponds to  the following objective 
function:
\begin{equation}
g(u,v) = ( u\circ v)^\top Q (u \circ v) - 2 (u\circ v)^\top l + 
         \lambda (  u^\top u  + v^\top v  )/2. 
\label{eqn:hpp_obj}
\end{equation}
Consideration of this parametrization 
and objective function may seem odd, as
the values of $u$ and $v$ beyond their element-wise product $\beta$ 
are not identifiable from the data.
However, $g$ is differentiable and biconvex, and its local minimizers 
can be found using a very simple 
alternating ridge regression algorithm.  
Furthermore, there is a correspondence between minimizers of 
$g$ and minimizers of $f$, which we state more generally as follows:
\begin{lemma} 
Let $f(\beta) = h(\beta) + \lambda || \beta||_1 $ and 
$g(u,v) = h( u\circ v) + \lambda( u^\top u + v^\top v)/2$. 
Then 
\begin{enumerate}
\item $\inf_\beta  f(\beta) =  \inf_{ u,v} g(u,v)$; 
\item if $(\hat u, \hat v)$ is a local minimum of $g$, then 
      $\hat\beta =\hat u \circ \hat v$ is a local minimum of $f$. 
\end{enumerate} 
\label{lemma:lasso_equiv} 
\end{lemma}
\begin{proof}
To show item 1 we write $u= \beta/v$,
where ``$/$'' denotes element-wise division,  so
that
\begin{align*}
\inf_{u,v} g(u,v) & = \inf_{\beta,v} g(\beta/v,v)  \\
 &= \inf_\beta \inf_v \left\{   h(\beta )  + \lambda 
  \left  (  ||\beta/v||^2  + ||v||^2\right  )/2  \right \} \\
&= \inf_\beta \left \{  h(\beta)  +
   \lambda \inf_v \left (  ||\beta/v||^2  + ||v||^2\right )/2  \right \} . 
\end{align*}
The inner infimum over $v$ 
is attained, and a minimizer $\tilde v$ can be 
can be found element-wise. The
$j$th element $\tilde v_j$ of a minimizer $\tilde v$ is simply
a minimizer of
$\beta_j^2 /v_j^2 + v_j^2$. If $\beta_j$ is zero then $\tilde v_j= 0$
is the unique global minimizer. Otherwise, this function is strictly
convex in $v_j^2$ with a unique minimum at $\tilde v_j^2 = |\beta_j|$.
The inner minimum  is therefore
\begin{align*} 
|| \beta/\tilde v||^2 + ||\tilde v||^2&  = \sum_{j=1}^p \left ( \beta_j^2 /\tilde v_j^2  + \tilde v_j^2 \right )  \\
&=    \sum_{j=1}^p \left  ( \beta_j^2/|\beta_j| + |\beta_j |  \right )
 = 2 || \beta ||_1, 
\end{align*}
and so 
\begin{align*} 
\inf_{u,v} g(u,v) 
 & = \inf_{\beta,v} g(\beta/v,v) \\
&= \inf_\beta \left \{  h(\beta) +
 \lambda \min_v\left  (  ||\beta/v||^2  + ||v||^2\right )/2   \right \} \\
&= 
 \inf_\beta \left \{ h(\beta )  + \lambda ||\beta||_1 \right \}  = \inf_\beta f(\beta) .  
\end{align*}
This proves item 1. In this proof, we saw
that the constrained minimum of $u^\top u + v^\top v$ 
subject to $u\circ v=\beta$ is  
attained when $u_j^2 = v_j^2 = |\beta_j|$. 
Since $h$ only depends on $u \circ v$, 
a local minimizer $(\hat u, \hat v)$ of  $g$ 
must also be a minimizer of 
$u^\top u + v^\top v$ subject to the constraint that
$u\circ v= \hat u \circ \hat v = \hat \beta$, and so 
$\hat u_j^2 = \hat v_j^2=|\hat \beta_j|$, giving a local 
minimum value of 
$g(\hat u, \hat v) = h( \hat \beta) + \lambda ||\hat \beta||_1 = 
  f( \hat \beta)$. 
That this must be a local minimum of $f$ follows from the 
fact that 
the image of any ball in $\mathbb R^{p}\times \mathbb R^p$ around $(\hat u,\hat v)$
under the mapping $(u,v)\rightarrow u \circ v$ contains a ball 
in $\mathbb R^p$ around $\hat u \circ \hat v$.  This proves item 2. 
\end{proof}

We now return to the definitions of $f$ and $g$ in 
Equations (\ref{eqn:l1_obj}) and 
(\ref{eqn:hpp_obj}),  where 
$h(\beta) = \beta^\top Q \beta - 2 \beta^\top l $. 
Since all local minimizers of $f$ are global minimizers 
\citep{tibshirani_2013},
item 2 of the lemma shows that 
any local minimizer $(\hat u, \hat v)$ of $g$ provides 
a global optimizer  $\hat \beta  = \hat u \circ \hat v$
of the lasso objective function. 
In other words, lasso estimates can be obtained 
from local minimizers of $g$.  
Such minimizers can be found with a simple and intuitive 
alternating ridge regression algorithm. 
To see this, rewrite $( u\circ v)^\top Q (u \circ v)$ as 
    $u^\top ( Q \circ  v v^\top ) u$, and 
    $(u\circ v)^\top l$ as $u^\top (v\circ l) $, so that
\[
 g(u,v) =  u^\top ( Q \circ  v v^\top + \tfrac{\lambda}{2} ) u + 2u^\top (v\circ l)  + \lambda v^\top v/2. 
\]
This is quadratic in $u$ for fixed $v$, 
with a unique minimizer  of
 $\tilde u = ( Q \circ  v v^\top + \tfrac{\lambda}{2} I  )^{-1} (l \circ v)$. 
Similarly, the unique minimizer of $g(u,v)$  in $v$  for fixed $u$ is 
 $\tilde v = ( Q \circ  u u ^\top + \tfrac{\lambda}{2} I  )^{-1} (l \circ u)$.
Iteratively optimizing $u$ and then $v$ given each other's current 
value 
is a type of coordinate descent algorithm. 
Since each conditional minimizer is unique, the algorithm 
will converge to a stationary point $(\hat u, \hat v)$ of $g$ \citep{luenberger_ye_2008}. 
At convergence, derivatives can be calculated to check if the 
point is a local minimizer (and therefore also a 
global minimizer). Alternatively, 
the optimality of
$\hat \beta= \hat u \circ \hat v$ 
can be
evaluated by checking if the Karush-Kuhn-Tucker (KKT) conditions
are approximately met:
Following \citet{tibshirani_2013},
the vector $\hat \beta$ is a global minimizer of
$f$ if
\begin{align}
2(l_j - [  Q \hat \beta]_j )/\lambda  &=  \text{sign}(\hat \beta_j )  \  \text{if} \ \hat\beta_j \neq 0  \label{eqn:kkt1},  \\
2(l_j - [  Q \hat \beta]_j )/\lambda & \in [-1,1]  \, \ \  \  \text{if} \  \hat \beta_j =0.  \label{eqn:kkt2} 
\end{align}
It is interesting to note that 
any stationary point $(\hat u,\hat v)$ of $g(u,v)$  will give 
a value $\hat \beta = \hat u \circ \hat v$ that satisfies
(\ref{eqn:kkt1}). To see this, note that at a critical point we have 
$(Q \circ \hat v \hat v^\top + \tfrac{\lambda}{2} I ) \hat u  = l\circ \hat v$,
which implies  
\begin{align*}
\hat u_j  &= 2 \hat v_j( l_j - [Q (\hat u\circ \hat v)]_j )/\lambda.
\end{align*}
Similarly, 
$\hat v_j = 2 \hat u_j ( l_j - [Q (\hat u\circ \hat v)]_j )/\lambda $.
If $\hat \beta_j = \hat u_j \hat v_j \neq 0$, then neither $\hat u_j$ nor  
$\hat v_j$ equal zero either, and so 
    $\hat u_j /\hat v_j = 2(l_j - [Q(\hat u \circ \hat v)]_j )/\lambda   
    = \hat v_j/\hat u_j. $
This implies that $\hat u_j^2=\hat v_j^2 $, or 
equivalently,
\begin{align*} 
2 (l_j - [ Q\hat \beta ]_j )/\lambda = 
\frac{\hat u_j}{\hat v_j} 
 &=     \frac{ \text{sign}(\hat u_j)}{ \text{sign}(\hat v_j)}   \\
 &  =  
  \text{sign}(\hat u_j \hat v_j) = 
\text{sign}(\hat \beta_j), 
\end{align*}
and so condition (\ref{eqn:kkt1}) is met. Not all stationary points
will satisfy (\ref{eqn:kkt2}), though. For example, the point 
$(0,0)\in \mathbb R^p \times \mathbb R^p$ is a stationary 
point of $g$ but is not a local minimum and does not satisfy 
 (\ref{eqn:kkt2}). 
However, in all of the numerical examples I have 
evaluated, 
the HPP algorithm has converged 
to objective function values that were as good or 
better than those of other algorithms.

\subsection{Numerical evaluation}
\label{ssec:l1numeric}
The HPP provides a  simple, intuitive algorithm for obtaining
lasso regression estimates.
Given
a starting value $v$ (such as one based on an OLS or ridge regression
estimate),
the algorithm is to
iterate steps 1 and 2 below until 
a convergence criteria is 
met:
\begin{enumerate}
\item Set $u = (Q \circ  vv^{\top} + \tfrac{\lambda}{2} I  )^{-1} (l \circ v)$;
\item Set $v = (Q \circ  uu^\top + \tfrac{\lambda}{2} I  )^{-1} (l \circ u)$.
\end{enumerate}
One justification of the HPP algorithm is that,
for some researchers, optimization of
$g$ via this HPP algorithm may be more intuitive and easier to code than
alternative optimization schemes for $f$ that require
an understanding of convex optimization.
With this in mind, it is of interest to compare the convergence of
the HPP algorithm to other intuitive and/or easy to implement
algorithms.
One such algorithm
is the local quadratic approximation (LQA) algorithm
of \citet{fan_li_2001}, which
also proceeds via iterative ridge regression. 
Specifically, one iteration of the LQA algorithm
is as follows:
\begin{enumerate}
\item Compute $D=\text{diag}( |\beta_1|^{-1} ,\ldots, |\beta_p|^{-1});$
\item Set $\beta=( Q  + \tfrac{\lambda}{2} D)^{-1} l$.
\end{enumerate}
The idea behind this algorithm is that $\beta^T  D\beta /2$ is a
quadratic approximation to the $L_1$ penalty
$||\beta||_1$ in a neighborhood around the current value of $\beta$.
This algorithm can equivalently be interpreted as an
EM algorithm for finding the posterior
mode of $\beta$ under independent Laplace prior
distributions on the elements of $\beta$
\citep{figueiredo_2003}.
However, inspection of the algorithm reveals a potential
problem: As entries of $\beta_j$ approach zero the
corresponding diagonal entries of $D$ approach infinity,
which
could lead to
numerical instability in the calculation of the
update to $\beta$ in step 2 of the algorithm.
In particular, the condition number of the
matrix $Q +  q \tfrac{\lambda}{2} D$
will generally approach infinity as
entries of $\beta$ approach zero.
\citet{hunter_li_2005} propose to remedy to this
potential numerical instability of LQA
by
perturbing the
update in step 2 so that
$D$ remains bounded.
In the context of $L_1$-penalized estimation,
this modification amounts to
replacing
$|\beta_j|^{-1}$,
the $j$th diagonal element of $D$,
with $(|\beta_j| + \epsilon)^{-1}$, where
$\epsilon$ is some small positive number, thereby ensuring that
the condition number of  $Q+\tfrac{\lambda}{2} D$  does not
go to infinity.
Note that in contrast,
the matrices $Q \circ vv^\top + \tfrac{\lambda}{2} I$
and $Q \circ uu^\top + \tfrac{\lambda}{2} I$  in
steps 1 and 2 of the HPP algorithm require no such modification, and
remain well-conditioned as elements of $u$ or $v$ approach
zero.
\begin{figure}
\centerline{\includegraphics[width=6in]{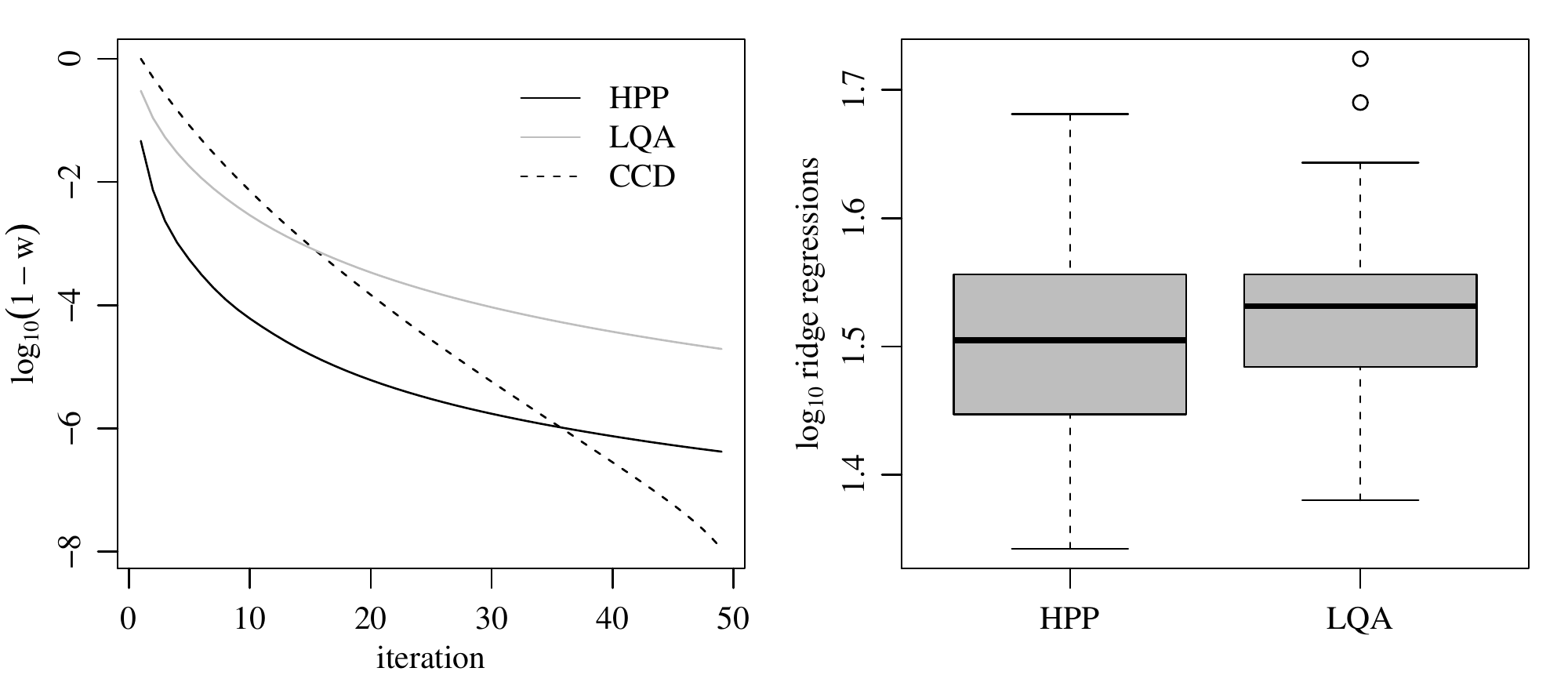}}
\caption{Simulation results for $L_1$-penalized regression.
The left panel shows the progress of each algorithm per
iteration, on average across the 100 datasets.
The right panel shows the variability in the number of
ridge regressions 
required until
the convergence criterion is met. }
\label{fig:lm_l1p}
\end{figure}
Another natural candidate for comparison to the HPP algorithm is
the ``shooting'' cyclic coordinate descent (CCD) algorithm
described by \citet{fu_1998}.
This algorithm optimizes the
objective function iteratively for each element of $\beta$
using directional derivatives.
While somewhat more complex than 
HPP and LQA in terms of  understanding and implementation,
CCD  does  not require matrix
inversions and is
perhaps the most popular algorithm for obtaining $L_1$-penalized
regression coefficients.

The convergence properties of the HPP, CCD and  $\epsilon$-perturbed  LQA
algorithms (with $\epsilon=10^{-12}$) were compared on 100
datasets that were simulated from the linear regression model
$y \sim N_n( X\beta ,  I)$.
A different value of $\beta$ was generated for each dataset,
with entries simulated independently from a
50-50 mixture of a point-mass at zero and
a mean-zero normal distribution with a standard deviation of 1/2.
For each dataset,
the entries
of the design matrix
$X$ were independently simulated from a standard normal distribution.
Results are presented here for the case that
$n=150$ and $p=100$. 
Other simulation scenarios may be explored using the replication code 
available at my website. 

For each simulated dataset,
a moment-based
empirical Bayes estimate of $\lambda$ was obtained and
all three algorithms were iterated 50 times, starting at
the unpenalized least squares estimate. 
Let $f( \beta_h^{(i)})$,
$f( \beta_l^{(i)} )$ and $f( \beta_c^{(i)})$
 denote the values of the objective
function at the $i$th iterate of the HPP, LQA and CCD algorithms,
respectively.
Letting
$f_{\max} = \max \{ f( \beta_h^{(1)}),f( \beta_l^{(1)}), f( \beta_c^{(1)})\}$
and
$f_{\min}=\min \{ f( \beta_h^{(50)}),f( \beta_l^{(50)}),f( \beta_s^{(50)})\}$,
the value of $w^i_h =  ( f_{\max} - f( \beta_h^{(i)} ) )/
               ( f_{\max} - f_{\min}  )  \in [0,1] $
measures the progress of the HPP algorithm at iterate $i$,
and $w_l^i$ and $w_c^{i}$  can be defined similarly for the  LQA and CCD
algorithms.
The upper-left panel of
Figure \ref{fig:lm_l1p} plots the values of
$w_h^i$, $w_l^i$ and $w_c^i$ for each iteration $i$, on average
across the 100 simulated datasets.
The HPP algorithm makes
substantially faster progress than either of the other two
algorithms initially.
However, 
HPP requires two ridge regressions per iteration,
whereas LQA requires only one.
To compare the computational costs of these algorithms,
we need to evaluate the number of iterations  
each requires to find a solution.
To this end, each of the three 
algorithms was iterated until
\begin{equation}
\max_{j}\left\{ (\beta^{(i)}_j - \beta^{(i+1)}_j)^2 \sum_{k=1}^n x_{k,j}^2\right \} 
  \leq \delta, 
\label{eqn:ccrit}
\end{equation}
 where $\delta$ was taken to be $\delta= 10^{-6}$. 
The left side of this inequality
is the convergence statistic used by the {\sf R}-package
{\tt glmnet} \citep{friedman_hastie_tibshirani_2010}. 

Objective functions and parameter estimates at convergence were compared
across all three algorithms, and there was no evidence of any substantial
differences: Relative differences in objective function values
were below 0.001\% $(10^{-5})$ for all 100 datasets, 
and relative squared differences
among parameter estimates  $\hat \beta$ and fitted values
$X\hat\beta$ were also all less than
0.001\%. 
The relative
mean-squared estimation error  $|| \hat\beta - \beta ||^2/||\beta||^2$ and the
mean-squared prediction error 
$|| X (\hat\beta - \beta) ||^2/|| X \beta||^2$
were 0.083
and 0.084, respectively, for all three algorithms.

The median numbers of iterations until convergence 
criterion (\ref{eqn:ccrit}) was met 
were 16, 34 and 29 for  HPP, LQA and CCD 
respectively.  
The variability in the number of iterations
until convergence for HPP and LQA is displayed in
the second panel of Figure \ref{fig:lm_l1p},
where for comparison the results are given in terms of
the number of
ridge regressions required until convergence.
Roughly speaking,
HPP takes twice as much time per iteration
but requires slightly less than half as many iterations to converge,
resulting in a small reduction in 
average computational costs relative to LQA. 

The computational 
costs of CCD are hard to compare to those of
HPP and LQA as CCD involves
different types of calculations at each iteration: 
CCD
updates each of the $p$ coefficients cyclically, whereas 
HPP and LQA 
update multiple parameters at once but require 
matrix inversions at each iteration. 
An informal comparison on
my desktop computer gave the  total run-time to convergence for
all 100 datasets being
0.96, 1.10 and 3.40 seconds for the HPP, LQA and CCD algorithms respectively, 
using convergence criteria (\ref{eqn:ccrit}) with $\delta=10^{-6}$. 
All algorithms were coded in the {\sf R} programming environment 
using no {\sf C} or {\sf FORTRAN} code. It is likely that the runtime 
of CCD would improve relative to the other methods if such code 
were used, as each iteration of CCD involves a for-loop over 
the elements of $\beta$ which is particularly slow in {\sf R}.

\subsection{HPP for sparse high-dimensional regression}  
While simple to explain and implement, the HPP algorithm 
requires two matrix inversions 
per iteration and so  becomes increasing computationally
costly as $p$ increases. 
Such  costs can be reduced by  
updating  $u$ and $v$ via Cholesky decompositions 
instead of matrix inversions 
(see the replication code for details), but these calculations 
still require $O(p^3)$ operations per decomposition. 

However, the structure of the HPP algorithm permits a 
modification that can substantially reduce 
computational costs in sparse high-dimensional 
settings. Recall that the HPP update for the vector 
$u$ is 
$u=(Q \circ vv^\top + \tfrac{\lambda}{2} I)^{-1} ( l\circ v)$, 
with an analogous update for $v$. 
If $v$ is sparse then  so is 
$l\circ v$, and the rows and columns of 
the matrix $(Q \circ vv^\top + \tfrac{\lambda}{2} I)$
can be reordered to yield a block diagonal 
matrix with one block being $\lambda/2$ times 
an identity matrix. 
From this we can see that the elements of the updated $u$-vector 
corresponding the zero elements of $v$ will also  
be zero,  while the remaining elements of 
$u$ will be given by 
$\tilde u = ( \tilde Q \circ \tilde v \tilde v^\top + \tfrac{\lambda}{2}I)^{-1}
 (\tilde l \circ \tilde v) ,
$
where $\tilde Q$, $\tilde l$, and $\tilde v$ are the 
submatrix and subvectors of $Q$, $l$, and $v$ 
corresponding to the non-zero elements of $v$. 
If the number of non-zero elements of $v$ is small, then 
the HPP update for $\tilde u$ (and thus $u$) can be 
computed quickly. 
Such a simplification is also 
possible for LQA in a limiting sense: 
If $D= \text{diag}(|\beta_1|^{-1},\ldots, |\beta_p|^{-1})$  then
the LQA update  matrix $(Q+\tfrac{\lambda}{2} D )^{-1}$ gets closer to being 
block diagonal 
as the elements of $\beta$ approach zero.  

Unfortunately, these algorithms cannot take 
advantage of sparsity because neither algorithm 
produces coefficient updates that are exactly zero. 
A simple way to overcome this limitation 
is to set parameter values to zero if they 
are less in absolute value than some  prespecified threshold. 
While this ad-hoc solution can induce exact sparsity, 
the caveat is that in doing so the algorithm may get 
trapped: 
Once an entry of $u$ or $v$ is set to zero, 
it remains zero for all iterations of HPP that follow.
An easy fix to this potential problem is to induce 
sparsity not by ad-hoc thresholding, but by performing 
a CCD step,  which can update a parameter from sparse to 
non-sparse and vice-versa. 
One version of such a mixed algorithm 
is to alternate HPP and CCD steps, resulting in what 
may be called an ``HPCD'' algorithm (Hadamard product, cyclic descent). 
Given a current value of $\beta$, two consecutive iterations of an 
HPCD algorithm proceeds as follows:
\begin{enumerate}
\item Update each element of $\beta$ iteratively with CCD.
\item Let $\tilde \beta$ be the nonzero values of $\beta$ and 
       $\tilde v$ be the square root of the absolute values of $\tilde \beta$. 
\begin{enumerate}
\item Set $\tilde u = ( \tilde Q \circ \tilde v \tilde v^\top + \tfrac{\lambda}{2}I)^{-1} (\tilde l \circ \tilde v)$;
\item Set $\tilde v = ( \tilde Q \circ \tilde u \tilde u^\top + \tfrac{\lambda}{2}I)^{-1} (\tilde l \circ \tilde u)$;
\item Set $\tilde \beta = \tilde u \circ \tilde v$. 
\end{enumerate} 
\end{enumerate}
Similarly, an ``LQCD'' algorithm (local quadratic approximation, cyclic descent)
may be constructed by alternately performing a CCD update and then 
an LQA update on the non-zero coefficients. 

The HPCD, LQCD and CCD algorithms were compared in a simulation 
study of 
100 datasets, simulated as before except now 
$p=1000$. 
Each algorithm was iterated 300 times on all 100 datasets, 
and the per-iteration 
convergence progress was averaged across datasets in the same 
manner as 
in the previous simulation study. The results,  displayed 
graphically in the left panel of Figure \ref{fig:bigp}, 
\begin{figure}
\centerline{\includegraphics[width=6in]{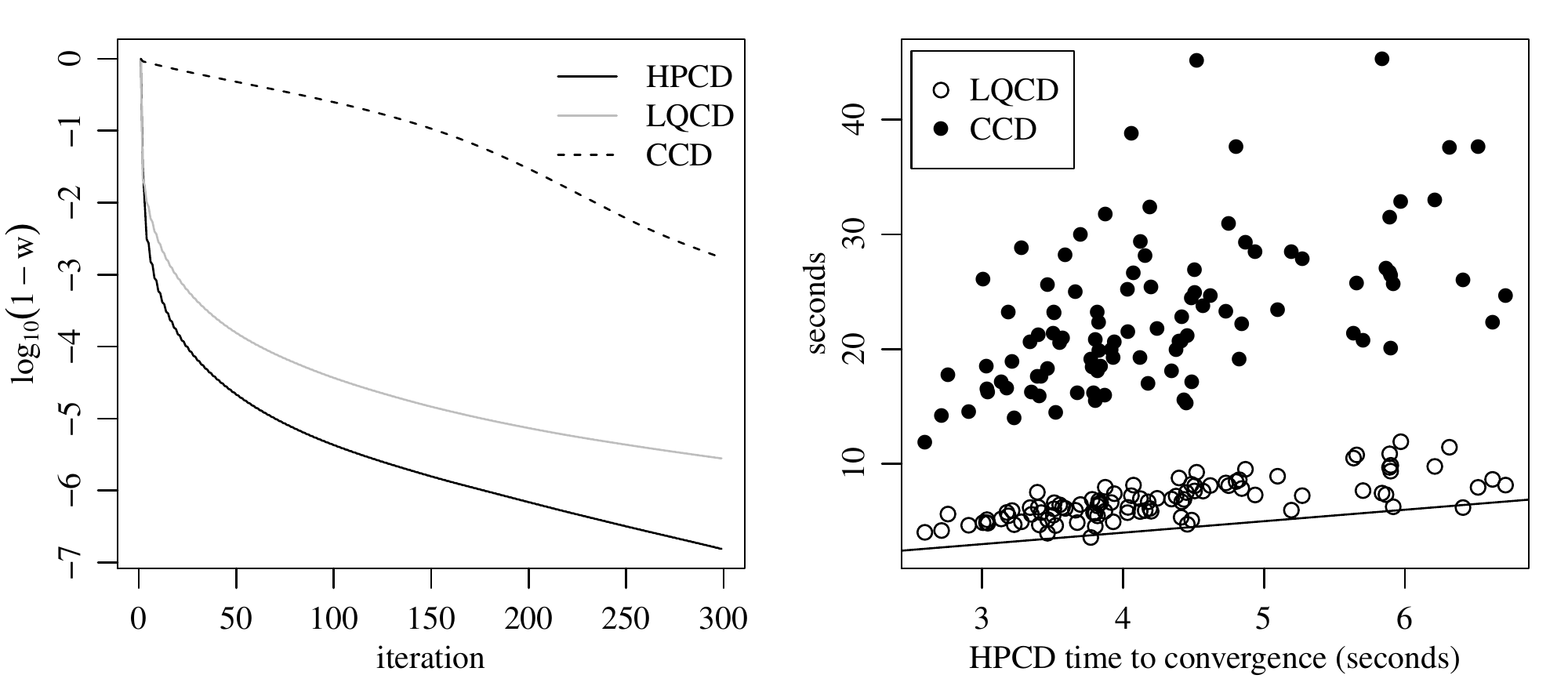}}
\caption{Simulation results for  $L_{1}$-penalized regression with $p>n$. 
The left panel shows the progress of each algorithm per
iteration, on average across the 100 datasets.
The right panel shows the time in seconds to convergence of 
each algorithm, across datasets.}
\label{fig:bigp}
\end{figure}
indicate that on a per-iteration basis HPCD 
faster than LQCD and 
substantially 
faster than CCD. 

The algorithms were also implemented on each dataset until 
convergence criteria (\ref{eqn:ccrit}) was met, with 
$\delta=10^{-6}$ as before. The three algorithms 
attained nearly the same objective function values as each other for 
each dataset, 
with relative differences being less than 0.001\% ($10^{-5}$)
across all datasets. 
Relative squared differences
among parameter estimates  $\hat \beta$ and fitted values
$X\hat\beta$ were less than
$1\%$ for all 
 simulated datasets.
Relative 
mean-squared estimation and prediction errors were 1.11 and 1.09 
respectively, for all algorithms. 
The median number of iterations until 
convergence were 
328, 649 and 1638 for the  HPCD, LQCD and CCD algorithms respectively. 
The computational costs per iteration 
of these algorithms 
are difficult to compare since the 
sizes of the matrices inverted by HPCD
and LQCD  vary with the sparsity 
level of the parameter estimate. 
However, the amount of computer time until convergence  
was recorded for each algorithm and dataset, and is displayed 
graphically in the right panel of Figure \ref{fig:bigp}. 
The HPCD algorithm was about 60\% faster than LQCD on average, 
and more than five times faster than CCD
(when implemented in {\sf R}). 

\subsection{Correlated predictors}
Since each iteration of CCD updates one element of $\beta$ at a time, 
its convergence properties may suffer if the Hessian 
of the objective function is not well conditioned.
This can occur if the columns of  $X$ 
are correlated
\citep{friedman_hastie_tibshirani_2010}. Conversely, HPP and LQA 
both update
entire vectors of parameters at once, and so their 
convergence properties may be robust to  
correlation among the predictors. We investigate this 
briefly with two simulation studies that are identical to  the previous two, 
except now each 
$X$ is a column-standardized version of 
the  random matrix $U  V^\top + E$ where 
$U\in \mathbb R^{n\times r}$, $V\in \mathbb R^{p\times r}$ and 
$E \in \mathbb R^{n\times p}$ are matrices with i.i.d.\ 
standard normal entries, with $r=p/10$. 
To see how this produces correlated predictors, note 
that for fixed $V$ the rows of  $U  V^\top + E$  
have covariance $VV^\top + I$. 

As before, the different algorithms were applied to 
each simulated dataset for a fixed 
number of iterations, and their per-iteration progress towards 
the optimal objective function was averaged across datasets. 
Results for the $p=100$ simulation study are shown in the left panel 
of Figure \ref{fig:corX}. Comparing this to the left panel of 
Figure \ref{fig:lm_l1p}, 
the relative convergence rate of CCD is substantially 
reduced as compared to the case of uncorrelated 
predictors. The median numbers of iterations to convergence 
were 35, 74 and 192 for HPP, LQA and CCD respectively, 
representing a roughly two-fold increase for HPP and LQA but 
more than a six-fold increase for CCD. 

The second panel of the figure compares HPCD, LQCD and CCD in the 
case that $p=1000$. The performance of CCD is 
much worse than HPCD and LQCD on a per-iteration basis, even 
more so than in the previous study with $p=1000$ (see Figure \ref{fig:bigp}). 
The median numbers of iterations to convergence for HPCD, LQCD and 
CCD for this simulation study were 240, 458 and 2427. 
This represents a decrease in the number of iterations until 
convergence for HPCD and LQCD as compared to the uncorrelated case, 
but the number of iterations needed by CCD is now roughly 10 times 
that of HPCD. 
Average time to convergence was 3.2, 4.9 and 32.0 seconds for 
HPCD, LQCD and CCD respectively, and so for this simulation scenario
HPCD is about 50\% 
faster than LQCD, and 10 times faster 
than CCD as implemented in {\sf R}.

\begin{figure}
\centerline{\includegraphics[width=6in]{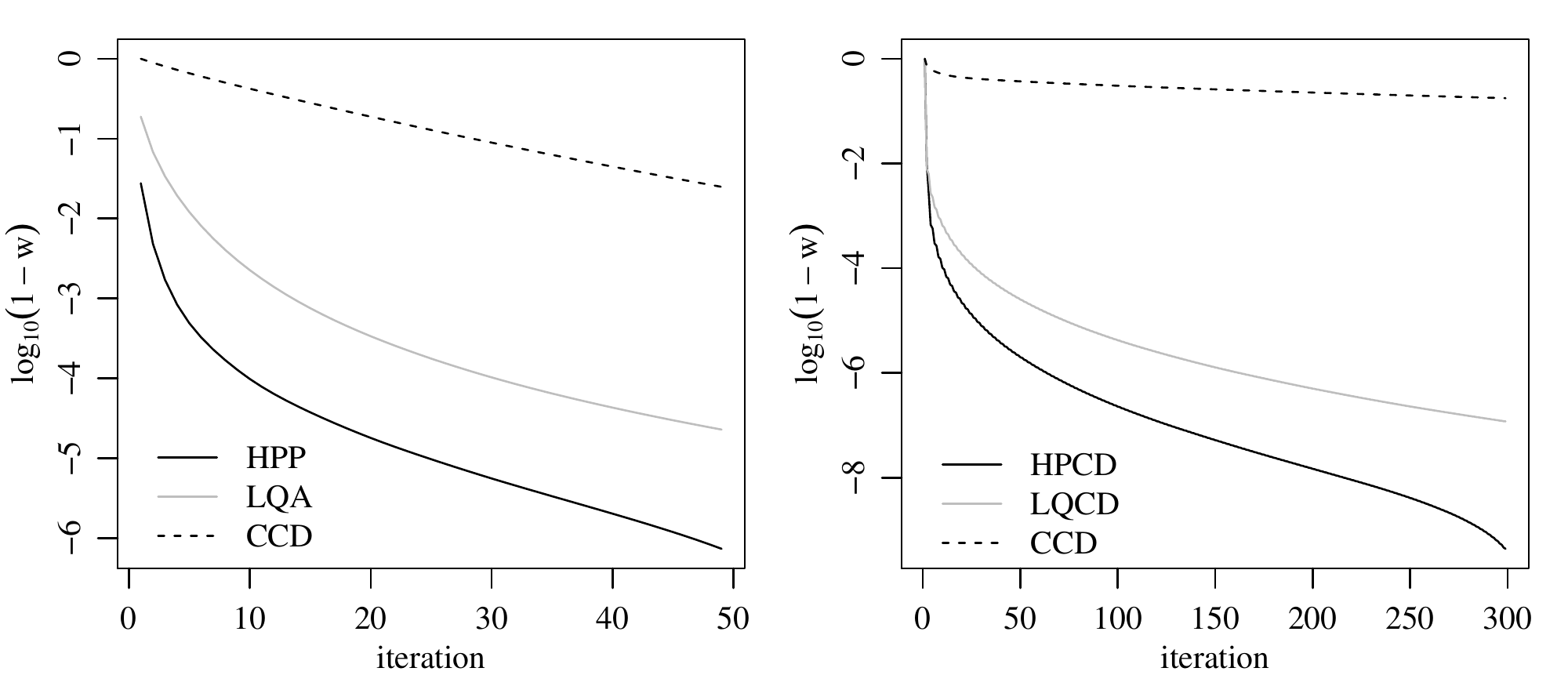}}
\caption{Average convergence progress for $L_{1}$-penalized regression 
algorithms with correlated predictors. The $p=100$ scenario is on the left, 
the $p=1000$ scenario is on the right. }
\label{fig:corX}
\end{figure}

\section{HPP for non-convex penalties}

\subsection{HPP for $L_q$-penalized linear regression}
A natural generalization of the  HPP is to write
$\beta = u_1 \circ \cdots \circ u_K $
and optimize
\begin{equation}
g(u_1,\ldots, u_K) = ( u_1 \circ \cdots \circ u_K )^\top Q  (u_1 \circ \cdots \circ u_K )- 2 ( u_1 \circ \cdots \circ u_K  )^\top l +  
   \tfrac{ \lambda}{K}\left (  u_1^\top u_1  + \cdots   u_K^\top u_K \right ). 
\label{eqn:khpp_obj}
\end{equation}
For $K=1$
the optimal $u$-value is the $L_2$-penalized
ridge regression estimate, and for $K=2$
the optimal value of  $(u_1,u_2)$ gives the
$L_1$-penalized lasso regression estimate, as discussed in the previous
section.
Values of $K$ greater than 2 correspond to
non-convex $L_q$ penalties with $q=2/K$.
For example, the $L_{1/2}$-penalized estimate is obtained by
optimizing (\ref{eqn:khpp_obj}) with $K=4$.
Non-convex penalties such as these have been studied by
 \citet{fan_li_2001},
\citet{hunter_li_2005},
\citet{kaban_durrant_2008},
\citet{zou_li_2008}, and
\citet{kaban_2013}
among others.
Such non-convex penalties induce sparsity
without the severe penalization
of large parameter values that is imposed by convex penalties,
such as the $L_1$ and $L_2$ norms.
Defining $||\beta||_q^q = \sum_j |\beta_j|^q$, the correspondence
between the $L_q$ penalties and the HPP
is given by the following lemma:
\begin{lemma}
Let $f(\beta) = h(\beta) + \lambda || \beta||_q^q $ where $q=2/K$, and let
 $g(u_1,\ldots, u_K)=h(u_1\circ \cdots \circ u_K)+ 
 \lambda( u_1^\top u_1 + \cdots u_K^\top u_K)/K$.
Then
\begin{enumerate}
\item $\inf_\beta  f(\beta) =  \inf_{u_1,\ldots, u_K} g(u_1,\ldots, u_K)$;
\item if $(\hat u_1,\ldots, \hat u_K)$ is a local minimum of $g$, then
   $\hat\beta =\hat u_1 \circ \cdots \circ\hat u_K$ is a local minimum of $f$.
\end{enumerate}
\label{lemma:fnp_equiv}
\end{lemma}
\begin{proof}
Using Lagrange multipliers, it is straightforward to show that the
minimum value of  $\sum_{k} u_k^\top u_k $ subject
to the constraint that $\prod_k u_{k,j} = \beta_j$ is attained
when $|u_{k,j}|$ is constant across $k$. This implies that
at a minimizing value, $u_{k,j}^{2} =  |\beta_j|^{2/K}$ and
that  $\sum_{k} u_k^\top u_k  = ||\beta||_{q}^q$, where
$q=2/K$.
The lemma then follows using exactly the same logic as used to prove
Lemma \ref{lemma:lasso_equiv}.
\end{proof}
The lemma  suggests a simple alternating ridge regression algorithm
for obtaining $L_q$-penalized linear regression estimates
in cases where $q=2/K$ for some integer $K$.
Given
a starting value $(u_2,\ldots, u_K)$,
 repeat  steps 1 and 2 below
 for each $k=1,\ldots,K$, iteratively
 until convergence:
\begin{enumerate}
\item Set $v= (u_1 \circ  \cdots  \circ u_K)/u_k$;
\item Set $u_k = (Q \circ v v^T + \tfrac{\lambda}{K} I)^{-1} (l\circ v) $.
\end{enumerate}
Note that $v$ in item 1 is simply the Hadamard product of the vectors $\{ u_1,\ldots, 
u_K\}$ except for $u_k$.
Another intuitive algorithm that is even easier to code is
the LQA algorithm for $L_q$-penalized linear regression, which
proceeds by iterating the following steps:
\begin{enumerate}
\item Compute $D=\text{diag}( |\beta_1|^{q-2} ,\ldots, |\beta_p|^{q-2});$
\item Set $\beta=( Q  + q\tfrac{\lambda}{2} D)^{-1} l$.
\end{enumerate}
As discussed in the previous section, the condition number of
$Q  + q\tfrac{\lambda}{2} D$
in step 2 of the algorithm will generally converge to
infinity as elements of $\beta$ approach zero, leading to
the potential for numerical instability.  The $\epsilon$-perturbed version
of this algorithm
proposed by  \citet{hunter_li_2005}
is to replace
$|\beta_j|^{q-2}$,
the $j$th diagonal element of $D$,
with $|\beta_j|^{q-2}\tfrac{|\beta_j|}{|\beta_j| + \epsilon}$, where
$\epsilon$ is some small positive number.
While this term no longer remains bounded as $|\beta_j|\rightarrow 0$
if $q<1$, it does
result in a minorization-maximization algorithm for
optimizing a perturbed  version of the objective function
(see \citet{hunter_li_2005} for details).
An alternative to LQA is the 
local linear approximation method (LLA) of 
\citet{zou_li_2008}. For $L_q$-penalized regression,
the LLA algorithm consists of iteratively
solving
an $L_1$-penalized regression problem as follows:
\begin{enumerate}
\item Compute $D=\text{diag}( |\beta_1|^{q-1} ,\ldots, |\beta_p|^{q-1});$
\item Set $\beta=\arg \min_\beta  \beta^\top Q\beta - 2 \beta^\top l + 
  2 q  \lambda  \sum_{j=1}^p d_j |\beta_j|. $
\end{enumerate}
As described by \citet{zou_li_2008}, this approach avoids the
potential numerical instability of the LQA algorithm, but requires
an $L_1$-optimization at each iteration.

\begin{figure}
\centerline{\includegraphics[width=6in]{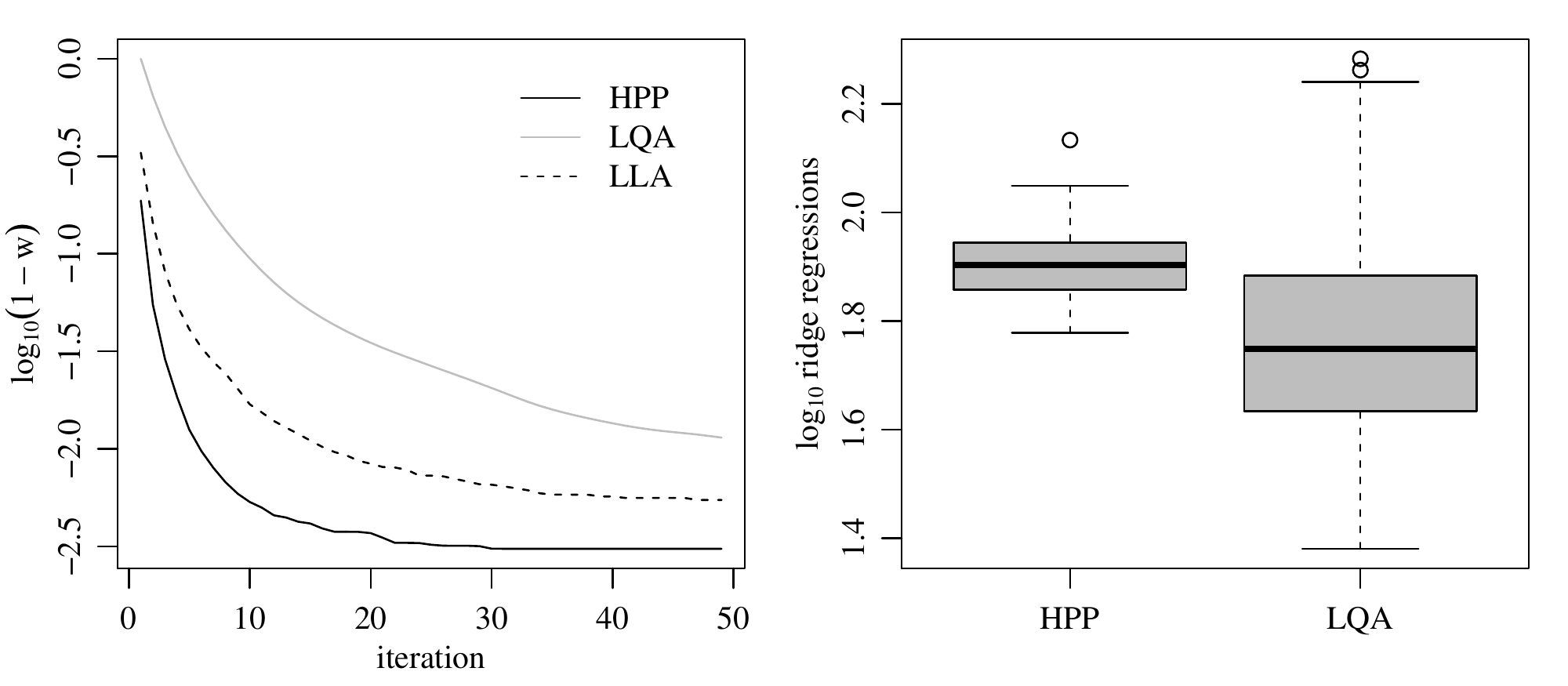}}
\caption{Simulation results for $L_{1/2}$-penalized regression.
The left panel shows the progress of each algorithm per
iteration, on average across the 100 datasets.
The right panel shows the variability in the number of
ridge regressions 
required until
the convergence criterion is met. }
\label{fig:lm_lqp}
\end{figure}

The convergence properties of HPP, LLA and $\epsilon$-perturbed LQA
were compared
on the
same simulated datasets as described in Section \ref{ssec:l1numeric},
but now using a nonconvex $L_{1/2}$  penalty corresponding to $K=4$.
A summary of the results are displayed in
Figure \ref{fig:lm_lqp}.
As can be seen in the left panel,
HPP converges at a faster rate per iteration
than LQA, on average across datasets.
Iterating until convergence, 
relative differences in objective function values
were all below .4\% $(0.004)$, 
and relative squared differences
among parameter estimates  $\hat \beta$ and fitted values
$X\hat\beta$ were less than
1\% for all but one of the simulated datasets. 
Relative mean-squared estimation error 
and 
 prediction
error
were 0.081
and 0.082, respectively, for all three algorithms.

The median numbers of iterations until convergence
were 20, 56 and 22 for the HPP, LQA and LLA
algorithms respectively.
Using the convergence criteria (\ref{eqn:ccrit}) with 
$\delta= 10^{-6}$, the LQA algorithm required
slightly more than three times (3.26) as many iterations as HPP
to converge,
on average across datasets.
However, HPP requires $K=4$
ridge regressions
per iteration as compared to
one for LQA. Taking this into account, the
two algorithms are comparable in terms
of computational burden,
as shown in the right panel
of Figure \ref{fig:lm_lqp}, 
with the LQA algorithm being slightly more efficient on average. 
While the computational costs of
HPP and LQA are easy to compare to each other,
the cost of LLA  is not easily comparable, as this algorithm involves
different operations than the other two and
depends on the particular
$L_1$-optimization method being used.
I compared the computational costs
of LLA to HPP and LQA in terms of their runtimes
on my desktop computer, using the {\sf R}-package {\tt penalized}
\citep{goeman_meijer_chaturvedi_2016} to implement the 
$L_1$-optimization required by LLA. 
The total elapsed time to convergence,
summed over all 100 simulated datasets,
was 3.3 and 2.2 seconds for HPP and LQA respectively, whereas
for LLA it was 15.0 seconds.
The relative lack of speed of the LLA algorithm is due to the fact that,
for these values of $n$ and $p$, performing an $L_1$-optimization
at each iteration is much more costly than performing the single
Cholesky factorization needed by HPP and LQA at each iteration.
For larger values of $p$ where Cholesky factorizations are more
costly, the LLA algorithm will be more competitive with
HPP and LQA.

\subsection{HPP for $L_q$-penalized generalized linear models}

The  arbitrariness of the 
function $h(\beta)$ in Lemmas \ref{lemma:lasso_equiv}
and  \ref{lemma:fnp_equiv} 
means that the HPP could be used
for penalized estimation in scenarios beyond linear regression. 
Consider a 
likelihood $L( \beta, \psi) = \prod_i p(y_i|\eta_i,\psi)$ where  
$\eta_i = \beta^\top x_i$ and 
$p(y_i | \eta_i ,\psi) = c(y_i) e^{\{(y_i \eta_i - A(\eta_i))/\psi \} }$, 
with $x_i$ being an observed 
$p$-variate vector of predictors for each observation $i$. 
The $L_q$-penalized likelihood estimate of $\beta$ is the maximizer 
of the penalized likelihood given by  
$\exp( -\lambda || \beta||_q^q/(2 \psi) ) \times 
  \prod_i p(y_i|\eta_i,\psi)$, 
or equivalently is the minimizer of two times the scaled negative log penalized 
likelihood 
\begin{align*}   
f(\beta)  =  2 \sum_i A(\beta^\top x_i) -2 
\beta^\top \sum_i y_i x_i 
          +  \lambda || \beta ||_q^q. 
\end{align*}
For example,  the special cases of 
linear regression, Poisson regression and logistic regression 
correspond to 
$A(\eta)$  being equal to  $\eta^2/2$, $e^{\eta}$ and 
$\ln(1+e^\eta)$  respectively.

By Lemma \ref{lemma:fnp_equiv}, 
if $q=2/K$ then the minimum of $f(\beta)$ is equal to the 
minimum of $g(u_1,\ldots, u_K)$, where 
\begin{align*}
g(u_1,\ldots, u_K) =  
  2 \sum_i A( [u_1\circ \cdots \circ u_K]^\top x_i)  - 2 
  (u_1\circ \cdots \circ u_K)^\top  X^\top y 
  + \frac{ \lambda}{K}\left (  u_1^\top u_1  + \cdots   u_K^\top u_K \right ). 
\end{align*}
Local minima of $g$ can be found using a variety of algorithms that 
iteratively update $u_1,\ldots, u_K$. 
For example, suppose we can optimize the following 
$L_2$-penalized 
likelihood function in 
$\beta$:
\[ f( \beta : X,\lambda ) =  2 \sum_i A(\beta^\top x_i) -2 
\beta^\top X^\top y 
          +  \lambda || \beta ||^2. 
\]
Now let $v= (u_1\circ \cdots \circ u_K)/u_k$. The part 
of $g$ that depends on $u_k$ can be written as 
\begin{align*}
g_k(u_k: X , v) &= 2 \sum_i A( [u_k\circ v]^{\top} x_i ) - 
  2 (u_k \circ v)^{\top} \sum_i y_i x_i + \tfrac{\lambda}{K} 
   u_k^\top u_k  \\
 &= 2 \sum_i A( u_k^\top [v\circ x_i] ) - 
    2 u_k^\top \sum_i y_i (v\circ x_i ) + \tfrac{\lambda}{K} 
   u_k^\top u_k  \\
  &=  f( u_k : \tilde X_k , \lambda/K ) , 
\end{align*}  
where $\tilde X_k = XD(v)$, so that 
the $i$th row of $\tilde X_k$ is $\tilde x_i = v\circ x_i$. 
If we can optimize $f(\beta:X,\lambda)$ in $\beta$, then 
we can optimize $g$ by iteratively optimizing 
$f( u_k : \tilde X_k , \lambda/K ) $ in $u_k$ for each $k=1,\ldots, K$. 
Therefore, any algorithm that provides 
$L_2$-penalized generalized linear model estimates 
can also be used to provide $L_q$-penalized estimates 
via the HPP, if $q=2/K$.  
For example,  $f(u_k: \tilde X_k , \lambda/K)$ 
can generally be optimized with the Newton-Raphson algorithm. The first and 
second derivatives of $f(u_k : \tilde X_k ,\lambda/K)$  are 
\begin{align*} 
d_k =  \frac{\partial f}{ \partial u_k }  = & 
  2 \left (   v \circ \sum x_i(\dot A (\eta_i)-y_i ) + \tfrac{\lambda }{K} u_k \right )  \\ 
H_k = \frac{\partial^2 f}{ \partial u_k \partial u_k^T }  &= 
2  \left ( v v^T \circ \sum \ddot A(\eta_i) x_i x_i^T   +  \tfrac{\lambda }{K}  I  \right  ), 
\end{align*}
where $\eta_i = (u_1\circ  \cdots \circ u_K)^T x_i$ and 
$\dot A$ and $\ddot A$ are the first and second derivatives 
of $A(\eta)$, respectively.
Critical points of $g$ can then be found by repeating the 
following steps iteratively 
for  $i=1,\ldots,K$ until 
convergence:
\begin{enumerate} 
\item Compute $v=( u_1\circ \cdots \circ u_K)/u_k$; 
\item Optimize $f(u_k: \tilde X_k ,\lambda/K)$ in $u_k$ 
     by iterating the following until convergence: 
\begin{enumerate}  
\item Compute $\eta_i = (u_k \circ v)^\top x_i$ for each $i=1,\ldots, n$, 
  then $d_k$ and $H_k$; 
\item Set $u_k$ to  $u_k - H_k^{-1} d_k $. 
\end{enumerate}
\end{enumerate}

The LQA strategy for $L_q$-penalized estimation in generalized linear models
is similar, except that at each iteration 
we optimize 
\[ 
2 \sum_i A(\beta^\top x_i) - 2 \beta^\top \sum y_ix_i + 
 q \lambda  \beta^\top D \beta, 
\]
where 
$  D =  \text{diag}( |\tilde \beta_1|^{q-2},\ldots, |\tilde \beta_p|^{q-2})$, 
with $\tilde \beta$ being the value of $\beta$ at the current iteration 
of the algorithm. 
The $\epsilon$-perturbed version of this LQA algorithm is, 
as in the previous subsection, obtained by replacing 
$|\beta_j|^{q-2}$,
the $j$th diagonal element of $D$,
with $|\beta_j|^{q-2}\tfrac{|\beta_j|}{|\beta_j| + \epsilon}$, where
$\epsilon$ is some small positive number. 
This perturbed LQA algorithm 
with $\epsilon=10^{-12}$ 
was compared to the HPP and LLA algorithms
in terms of obtaining $L_{1/2}$-penalized
estimates of logistic regression coefficients. 
For each of the 100 values of $\beta$ and $X$ 
used in the previous simulation study, a vector 
of $n=150$ binary observations were independently simulated 
from the 
logistic regression model
$y_i \sim \text{binary}( e^{\beta_i^\top x_i}/(1+e^{\beta_i^\top x_i}))$,
for $i=1,\ldots, n$. 
As shown in the left panel of Figure \ref{fig:glm_lqp}, 
each iteration of
HPP provides a larger improvement to the objective function 
than an iteration of either LLA or LQA, on average across datasets. 
Each of the three algorithms were also iterated until convergence 
criteria (\ref{eqn:ccrit}) was met. Unlike for the previous three simulation 
studies, the objective functions at convergence were sometimes 
non-trivially different: 
Objective functions differed by as much as 4.3\%, and 
differed by 1\% or more for 18 of the 100 datasets. 
The HPP algorithm attained a lower (better) objective function 
than the LQA algorithm for 88 of the 100 datasets, and 
a lower objective function than LLA for 60 datasets. 
However, 
 even though these algorithms produced solutions
with non-trivial objective function differences, their
estimation accuracies were nearly identical:
Relative
MSE $|| \hat\beta - \beta ||^2/||\beta||^2$ and the 
mean function estimation error
$|| X (\hat\beta - \beta) ||^2/|| X \beta||^2$
were both about 0.78
on average across datasets for all three methods.

While 
the number of iterations to convergence for the LQA 
algorithm was nearly five times (4.65) that of the HPP algorithm 
on average, 
HPP requires $K=4$ Newton-Raphson optimizations per iteration
whereas LQA only requires one.  As shown in the right panel 
of the figure, this results in 
the HPP algorithm being slightly less computationally costly 
than the LQA algorithm, on average. 
In term of comparison to LLA, 
while for linear regression  the computational costs 
per iteration of HPP and LQA were much lower than that of LLA, 
for logistic regression the converse is true:
In terms of elapsed times to convergence, 
HPP and LQA took a total of 27.5 and 30.2 seconds respectively
to obtain estimates 
for all 100 datasets, whereas LLA took only  14.7 seconds -
about twice as fast for this particular simulation scenario. 
This reversal is because all 
three logistic regression algorithms 
involve iterative optimization schemes
within each iteration (IRLS for HPP and LQA, gradient descent for LLA), 
whereas for linear regression this was the case only for LLA.

\begin{figure}
\centerline{\includegraphics[width=6in]{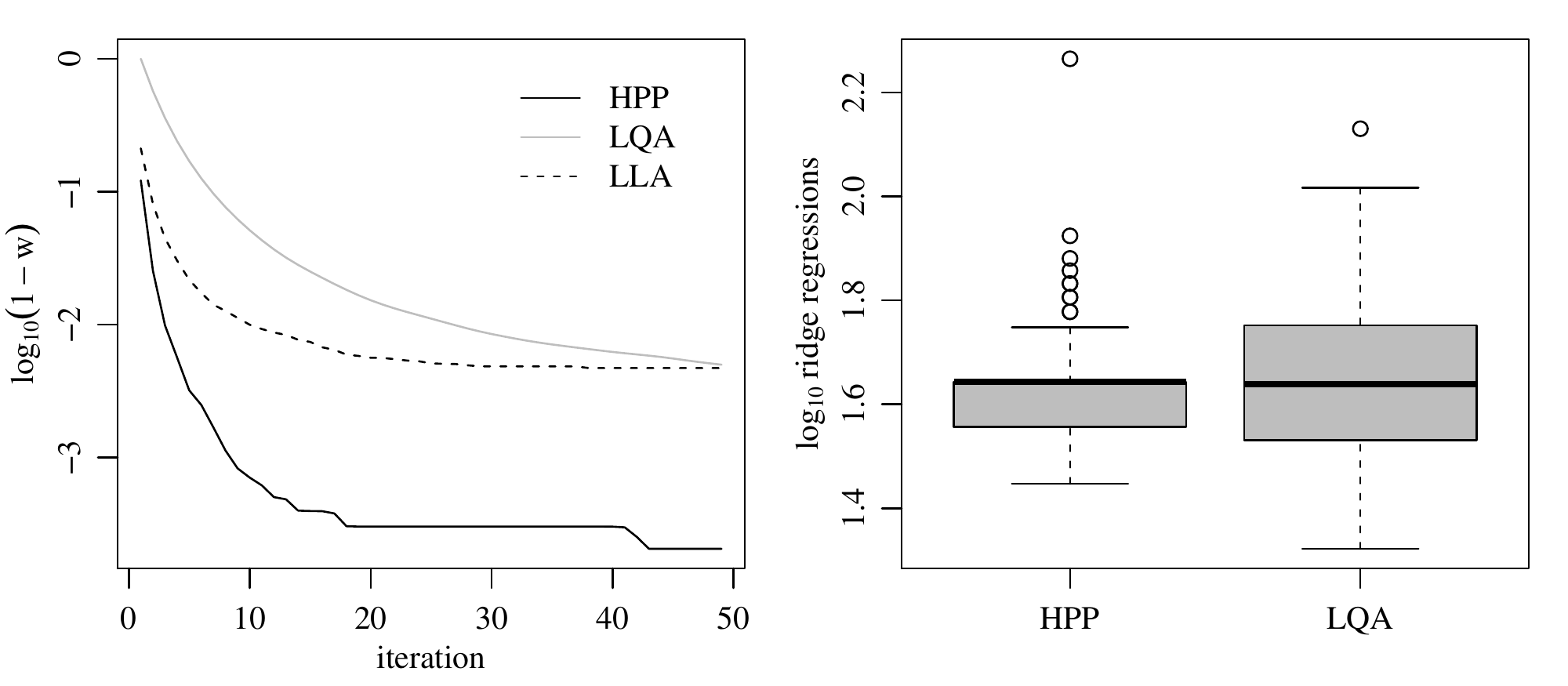}} 
\caption{Simulation results for $L_{1/2}$-penalized logistic regression.
The left panel shows the progress of each algorithm per
iteration, on average across the 100 datasets.
The right panel shows the variability in the number of
ridge regressions 
required until
the convergence criterion is met. }
\label{fig:glm_lqp}
\end{figure}

\section{Structured penalization with the HPP}  
It is well-known that the lasso objective function
$f(\beta)= \beta^T Q \beta - 2 \beta^T l + \lambda||\beta||_1$
is equal to the scaled log posterior density of $\beta$ 
under a Laplace prior distribution on the elements of
$\beta$ \citep{tibshirani_1996, 
figueiredo_2003, 
park_casella_2008}. 
Specifically, for the linear regression model 
$y\sim N(X\beta, \sigma^2 I )$  and prior  distribution   
$\beta_1,\ldots, \beta_p \sim$ i.i.d.\ Laplace$(\lambda/[2\sigma^2])$, 
the posterior density of $\beta$ is given by 
\begin{align*}
p(\beta | y, X,\sigma^2 ) & \propto 
  \exp( -||y-X \beta ||^2 /[2\sigma^2 ] ) 
  \exp( - \lambda ||\beta||_1/[2\sigma^2 ] ) \\
&  \propto 
  \exp( -\tfrac{1}{2\sigma^2} [ \beta^\top Q \beta - 2 \beta^\top l  + \lambda ||\beta||_1 ] )  = \exp( -\tfrac{1}{2\sigma^2} f(\beta) ) , 
\end{align*}
where $Q= X^\top X $ and $l = X^\top y$ as before. 
The lasso estimate $\hat \beta$ is therefore equal to 
the  posterior mode estimate.
Alternatively, reparametrizing
the regression model so that $\beta= u\circ v$, 
and using independent 
$N(0, 2\sigma^2/\lambda)$ prior distributions
 for the elements of $u$ and $v$ gives
\begin{align*}
p(u,v | y, X,\sigma^2 ) & \propto 
  \exp( -||y-X (u\circ v) ||^2 /[2\sigma^2 ] ) 
  \exp( - \tfrac{1}{2\sigma^2} [  \lambda u^\top u/2 + \lambda v^\top v/2 ]) \\
&  \propto 
  \exp( -\tfrac{1}{2\sigma^2}  g(u,v) ), 
\end{align*}
where $g(u,v)= (u\circ v)^T Q (u\circ v) - 2 (u\circ v)^T l + 
  \lambda( u^\top u+ v\top v)/2$. 
The minimizers $\hat u$ and $\hat v$ of $g(u,v)$ may therefore 
be viewed as posterior mode estimates under independent 
Gaussian priors on $u$ and $v$, with the lasso estimate 
given by $\hat \beta = \hat u \circ \hat v$. 

The $L_1$ and $L_2$ penalties
(or Laplace and Gaussian priors)  
on $\beta$ and $(u,v)$
respectively induce sparsity in the 
parameter estimates, but in an unstructured way. 
From a Bayesian perspective, the
\emph{a priori}
independence of the parameters 
means that the 
parameters convey 
no information about each other. In particular, 
the shrinkage of any one parameter is 
unrelated to that of any other. 
However, 
in many estimation problems there are relationships among the elements of 
$\beta$, and 
it may be desirable to shrink related parameters 
by similar  amounts or towards a common value. 
For example, a subset of elements 
of $\beta$ may correspond to the effects of different 
levels of a single categorical predictor. 
For models with such variables, 
the group lasso penalty of \citet{yuan_lin_2006}
may
shrink the entire 
subset to zero, 
in which case it would be inferred that  there is 
no effect of the categorical predictor. 
In other situations, the elements of $\beta$ may represent variables 
that have  spatial or temporal locations. 
To estimate such parameters, 
\citet{tibshirani_saunders_rosset_zhu_knight_2005} introduced the 
fused lasso, which in addition to penalizing 
the magnitudes of the elements of $\beta$, also penalizes 
the differences between elements that are spatially or 
temporally close to one another. 

These and other structured penalizations employ a variety 
of optimization techniques to obtain parameter estimates. 
As an alternative to these approaches, the 
HPP can be used to generate a class of structured sparse estimates
that can be obtained with a simple and intuitive
alternating ridge regression algorithm. 
Consider the objective function 
\begin{equation}
g(u,v) = (u\circ v)^T Q (u\circ v) - 2 (u\circ v)^T l +  
         u^\top \Sigma_u^{-1} u + v^\top \Sigma_v^{-1} v , 
\label{eqn:shpp} 
\end{equation}
where $\Sigma_u$ and $\Sigma_v$ are positive definite (covariance)  matrices. 
This objective function is equal to the scaled log-posterior density of 
$( u,v)$ under the model $y\sim N(X\beta ,\sigma^2 I)$ and 
independent prior 
distributions 
 $u/\sigma \sim N(0 , \Sigma )$, 
 $v/\sigma \sim N(0 , \Sigma )$. 
We refer to this combination of parametrization and 
penalty as a structured HPP, or SHPP. 
Note that the unstructured HPP corresponding to the $L_1$ penalty 
on 
$\beta = u\circ v$ is obtained by setting 
$\Sigma = \tfrac{2}{\lambda} I$. 

Local minima of $g(u,v)$ may be obtained 
with the following algorithm 
in which $u$ and $v$ are 
iteratively optimized  until convergence:
\begin{enumerate}
\item Set $u = ( Q \circ  vv^{\top} + \Sigma_u^{-1} )^{-1} (l \circ v)$; 
\item Set $v = ( Q \circ  uu^\top + \Sigma_v^{-1}  )^{-1} (l \circ u)$.
\end{enumerate}
In addition to the simplicity of the estimation approach, 
the SHPP also benefits from having 
a penalty that is easily interpretable as a 
covariance model for the relationships among the 
entries of $u$ and $v$,  and therefore among the entries of 
$\beta=u\circ v$.  
From a Bayesian perspective, 
if $u/\sigma\sim N(0 ,\Sigma_u)$ and 
   $v/\sigma\sim N(0 ,\Sigma_v)$, then 
  $ \Cov{ \beta/\sigma^2 }  = \Sigma_u \circ \Sigma_v$ 
(although $\beta$ is not \emph{a priori} Gaussian). 
Furthermore, 
since any positive definite matrix can be written 
as the Hadamard product of two other positive definite 
matrices \citep{majindar_1963,styan_1973}, 
$\Sigma_u$ and $\Sigma_v$ can always be chosen to 
yield a particular value of 
  $ \Cov{ \beta/\sigma^2 }$.

To illustrate the SHPP methodology, we analyze 
spatial data from a diffusion tensor imaging (DTI) study  
that compared the brain activity of 6 dyslexic children 
to that of  6 non-dyslexic controls, as described 
in 
\citet{deutsch_et_al_2005}. 
Following 
\citet{efron_2010}, we analyze  $z$-scores obtained 
from two-sample tests performed  at each of $p=15,443$ spatially 
arranged voxels. Each voxel has a location in a $73\times 55 \times 20$
three-dimensional grid. 
A plot of these $z$-scores at three adjacent vertical locations is 
given in Figure \ref{fig:brainraw}. 
\begin{figure}
\centerline{\includegraphics[width=6in]{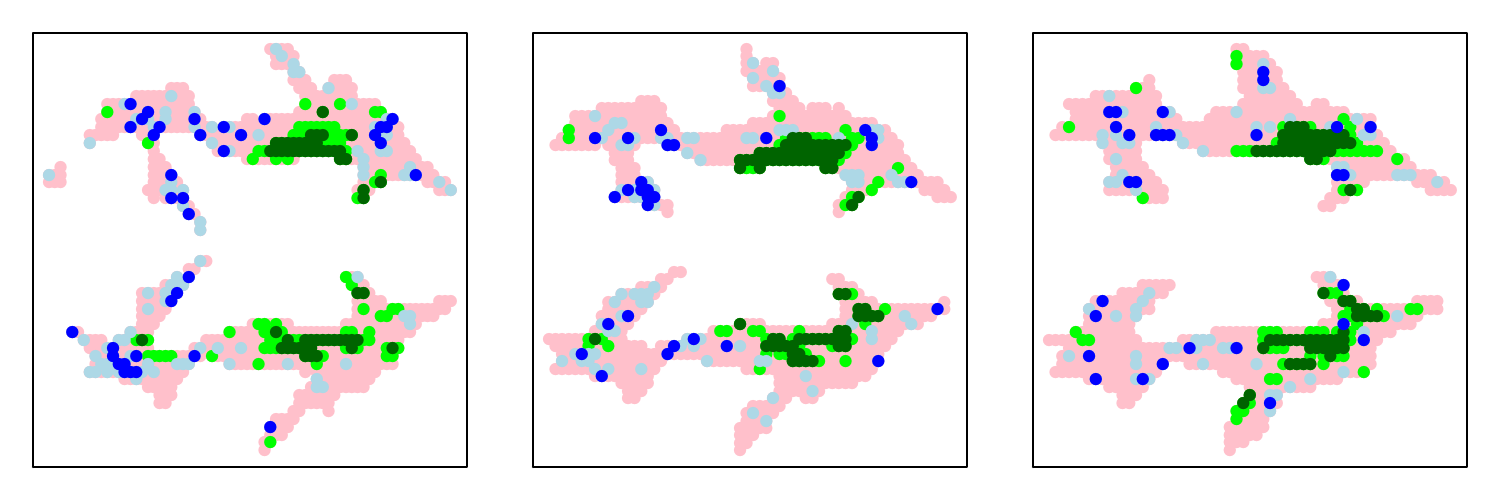}}
\caption{DTI $z$-scores at three 
adjacent vertical coordinates. 
Green and blue 
correspond to positive and negative values 
exceeding $z_{.975}$ in absolute value, 
and light green and light blue correspond to 
 values exceeding $z_{.9}$ in absolute value. }
\label{fig:brainraw}
\end{figure} 
The data exhibit a high degree of positive spatial dependence, 
with the 
$z$-values of 
neighboring voxels frequently being of the same sign.    
Also note that the spatial dependence among the large 
negative values appears to be much smaller than that of the 
large positive values, suggesting that many of the 
large negative values are due to noise.  

Possible causes of this dependence include 
spatially dependent measurement 
errors and spatially structured signals. 
While it is likely that both of these factors are 
contributing to the spatial dependence, 
for the purpose of illustrating the SHPP we 
assume that it is exclusively the latter. We model
these data as $z \sim N_p( \theta , I)$, and use 
a spatial SHPP for estimating $\theta$. 
Specifically, we write $\theta= u\circ v$ and optimize 
\begin{equation}
 g(u,v) = || z - u \circ v ||^2 +  u^\top \Sigma^{-1} u  + 
   v^\top \Sigma^{-1} v,  
\label{eqn:gbrain}
\end{equation}
where $\Sigma$ is the covariance matrix of a spatial conditional 
autoregression (CAR) model. The CAR model is parametrized 
as $\Sigma =  \tau^2 ( I - \rho G)^{-1}$, where 
$\rho$ and $\tau^2$ are parameters and 
$G$ is a matrix of spatial weights. 
Letting $n_i$ be the number of voxels spatially adjacent to 
voxel $i$, 
the weights are given 
by $g_{i,j} = 1/n_i$ if $j$ is adjacent to $i$ and $g_{i,j}=0$ 
otherwise. 
Under this model, the conditional expectation of
$u_i$ given $\{u_j:j\neq i\}$  is $\rho$ times the average value 
of its neighbors, and the conditional variance is $\tau^2$. 

Empirical Bayes estimates of $\rho$ and $\tau^2$ were obtained
from the data and then 
used to define the  SHPP objective 
function (\ref{eqn:gbrain}).  
To avoid 
calculations involving the $15,443\times 15,433$ covariance matrix 
$\Sigma$ that are
 necessary for the algorithm described above, we instead  
use a block coordinate descent algorithm that iteratively 
updates the 
values of $u_i$ and $v_i$ for each voxel $i$ as follows:
\begin{enumerate}
\item Compute $q=v_i^2 +1/\tau^2$ and $l = z_i v_i + \rho \bar u_{n_i}/\tau^2$. 
      Set $u_i = l/q$; 
\item Compute $q=u_i^2 +1/\tau^2$ and $l = z_i u_i + \rho \bar v_{n_i}/\tau^2$.
      Set $v_i = l/q$. 
\end{enumerate}
In the above algorithm, $\bar u_{n_i}$ denotes the 
average of the $u$-values among the $n_i$ neighbors of voxel $i$, 
and $\bar v_{n_i}$ is defined analogously. 
Starting with values of 
$u_i=|z_i|$ and 
$v_i = z_i/|z_i|$, this algorithm was iterated until the 
relative change in $u\circ v$  from one complete iteration 
over all voxels to the next 
was less than 
$10^{-10}$. 
This required 123 iterations and a 
little under two minutes on
my desktop computer. 
The resulting estimate
$\hat \theta$  is  very sparse, 
with about 94\% of the entries being less than 
$10^{-6}$ in absolute value (10 times smaller than the 
smallest  entry of $z$). 
As shown in 
the top row of Figure \ref{fig:brainshrink}, this sparsity is 
highly spatially structured, and 
a few large multi-voxel regions of the brain
with consistently positive  values
are identified. 
The lack of non-zero  negative values of $\hat \theta$ is 
in agreement with an analysis by \citet[chapter 4]{efron_2010}
using false discovery rates, and further suggests that 
the few large and negative raw-data values at spatially isolated 
voxels are the result of noise. 
For comparison, the second row  of the figure
summarizes 
an unstructured lasso estimate, 
where the shrinkage parameter $\lambda$ was obtained
using the same empirical Bayes approach described in 
Section 2.2.  
The unstructured estimate  
has non-zero negative values for 
several spatially isolated voxels, 
and the estimated positive regions 
are not as spatially coherent as those of 
the SHPP estimate.

\begin{figure}
\centerline{\includegraphics[width=6in]{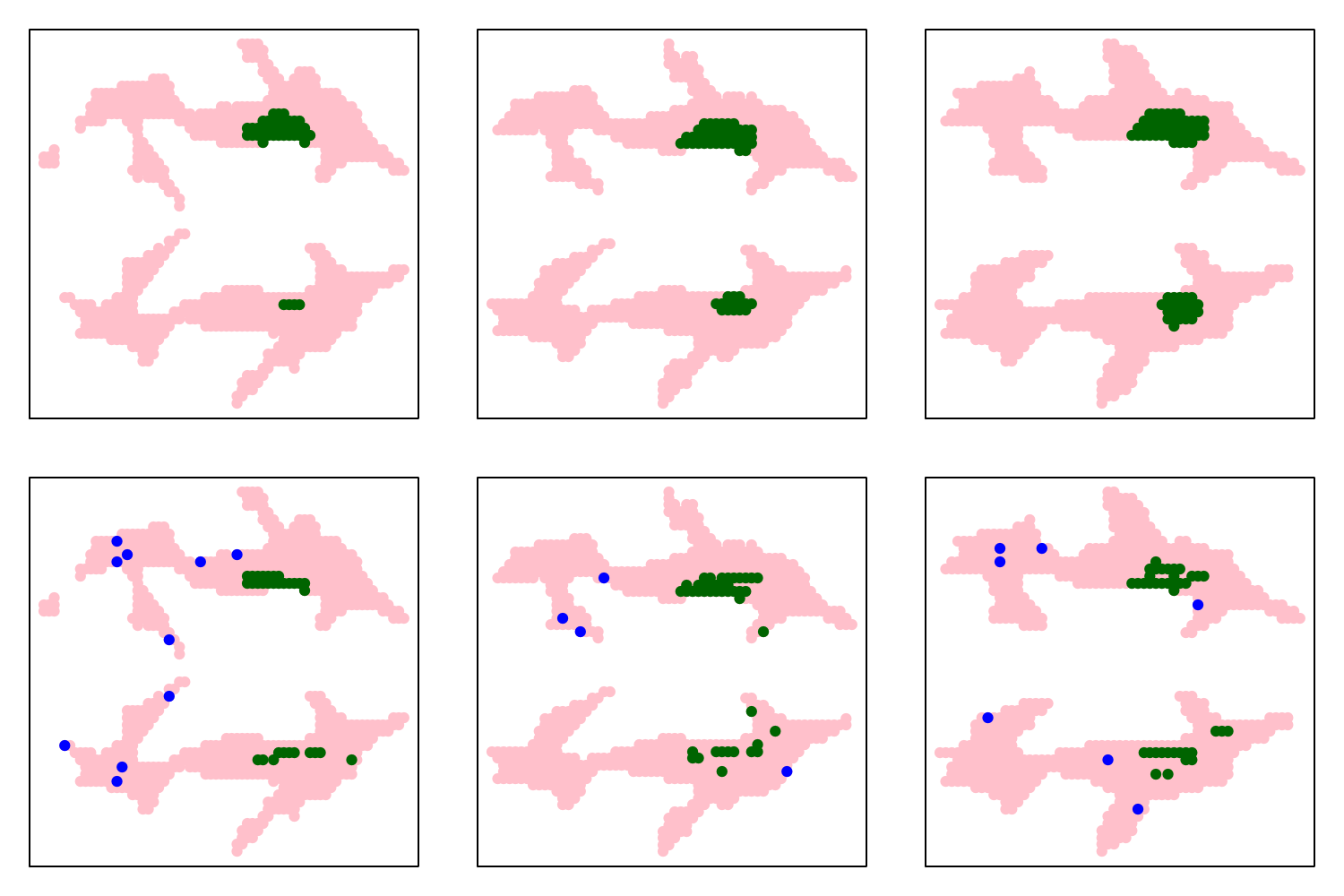}} 
\caption{
Estimated group differences using 
spatially structured (top) and unstructured (bottom) 
penalties. Green, blue and pink 
correspond to estimates that are  positive, negative and zero 
respectively. }
\label{fig:brainshrink}
\end{figure}

\section{Discussion} 
The Hadamard product parametrization provides 
a simple and intuitive method for obtaining 
$L_q$-penalized regression estimates 
for certain values of $q$. 
In terms of accessibility to practitioners, the 
HPP algorithm is similar to the LQA algorithm. 
Both of these algorithms proceed by iterative ridge regression. 
Unlike the ``ridge'' of the LQA algorithm, 
that of the HPP algorithm is 
bounded near zero, suggesting that HPP is to 
be preferred over LQA for reasons of numerical stability.
However, for the numerical examples in this article 
(and others considered by the author), instability was not 
an insurmountable issue for LQA and the two algorithms performed 
comparably. 

The HPP algorithm updates only non-sparse parameter values, 
meaning that its computational costs are greatly reduced 
when the parameter values are highly sparse. 
Sparsity in parameter values can be introduced by combining 
HPP with a cyclic coordinate descent (CCD) algorithm. In a 
simulation study with a large number of predictors ($p=1000$), 
the resulting 
hybrid algorithm exhibited extremely fast convergence 
relative to standard CCD. 
However, in non-sparse high-dimensional scenarios, HPP as presented 
here may be impractical as it requires a Cholesky factorization 
at each iteration. One possible modification of HPP for such cases 
would be to use a first-order method (e.g.\ gradient descent) 
to do the optimization at each iteration.

The $L_2$ penalties on the parameters in the HPP 
can be thought of as isotropic normal prior distributions. 
Similarly, non-isotropic quadratic penalties can be constructed 
that can be interpreted as Gaussian models for the parameters. 
Such penalties can be useful in situations where the 
relationships between the parameters are naturally 
expressed in terms of a covariance model. 
However, the analogy to Bayesian estimation is limited: 
The value $(\hat u,\hat v)$ that minimizes
$||y - X (u\circ v)||^2+  \lambda( ||u||^2  + ||v||^2 )/2$   is the posterior
mode of $(u,v)$ under isotropic 
Gaussian priors, but $\hat\beta = \hat u\circ \hat v$ is 
not the posterior mode of $\beta$ under the prior 
induced by the Gaussian priors on $u$ and $v$. 
This is because, in general, the 
posterior mode of a function of a parameter is not 
the function at the posterior mode of the parameter. 
If $u$ and $v$ are \emph{a priori} Gaussian
then the induced prior distribution on $\beta$ is not 
a Laplace distribution (which would yield the Lasso 
estimate as a posterior mode), but a 
``normal product'' distribution  \citep{weisstein_2016_npd}. 
This prior, considered by \citet{zhou_liu_fang_2015}, 
corresponds to a different penalty than any of the 
$L_q$ penalties, and can be shown to be 
in the 
class of normal-gamma prior distributions studied
by \citet{griffin_brown_2010}. 

\medskip

Replication code for the numerical examples in this article 
is available at my  website. I thank 
Panos Toulis for a helpful discussion.
This research was supported by NSF grant DMS-1505136.


\bibliography{hpp}

\begin{thebibliography}{}

\bibitem[\protect\citeauthoryear{Deutsch, Dougherty, Bammer, Siok, Gabrieli,
  and Wandell}{Deutsch et~al.}{2005}]{deutsch_et_al_2005}
Deutsch, G.~K., R.~F. Dougherty, R.~Bammer, W.~T. Siok, J.~Gabrieli, and
  B.~Wandell (2005).
\newblock Correlations between white matter microstructure and reading
  performance in children.
\newblock {\em Cortex\/}~{\em 41\/}(3), 354--363.

\bibitem[\protect\citeauthoryear{Efron}{Efron}{2010}]{efron_2010}
Efron, B. (2010).
\newblock {\em Large-scale inference}, Volume~1 of {\em Institute of
  Mathematical Statistics (IMS) Monographs}.
\newblock Cambridge University Press, Cambridge.
\newblock Empirical Bayes methods for estimation, testing, and prediction.

\bibitem[\protect\citeauthoryear{Fan and Li}{Fan and Li}{2001}]{fan_li_2001}
Fan, J. and R.~Li (2001).
\newblock Variable selection via nonconcave penalized likelihood and its oracle
  properties.
\newblock {\em J. Amer. Statist. Assoc.\/}~{\em 96\/}(456), 1348--1360.

\bibitem[\protect\citeauthoryear{Figueiredo}{Figueiredo}{2003}]{figueiredo_2003}
Figueiredo, M. A.~T. (2003, September).
\newblock Adaptive sparseness for supervised learning.
\newblock {\em IEEE Trans. Pattern Anal. Mach. Intell.\/}~{\em 25\/}(9),
  1150--1159.

\bibitem[\protect\citeauthoryear{Friedman, Hastie, and Tibshirani}{Friedman
  et~al.}{2010}]{friedman_hastie_tibshirani_2010}
Friedman, J., T.~Hastie, and R.~Tibshirani (2010).
\newblock Regularization paths for generalized linear models via coordinate
  descent.
\newblock {\em Journal of Statistical Software\/}~{\em 33\/}(1), 1--22.

\bibitem[\protect\citeauthoryear{Fu}{Fu}{1998}]{fu_1998}
Fu, W.~J. (1998).
\newblock Penalized regressions: the bridge versus the lasso.
\newblock {\em J. Comput. Graph. Statist.\/}~{\em 7\/}(3), 397--416.

\bibitem[\protect\citeauthoryear{Goeman, Meijer, and Chaturvedi}{Goeman
  et~al.}{2016}]{goeman_meijer_chaturvedi_2016}
Goeman, J.~J., R.~J. Meijer, and N.~Chaturvedi (2016).
\newblock {\em Penalized: L1 (lasso and fused lasso) and L2 (ridge) penalized
  estimation in GLMs and in the Cox model}.
\newblock R package version 0.9-47.

\bibitem[\protect\citeauthoryear{Griffin and Brown}{Griffin and
  Brown}{2010}]{griffin_brown_2010}
Griffin, J.~E. and P.~J. Brown (2010).
\newblock Inference with normal-gamma prior distributions in regression
  problems.
\newblock {\em Bayesian Anal.\/}~{\em 5\/}(1), 171--188.

\bibitem[\protect\citeauthoryear{Hunter and Li}{Hunter and
  Li}{2005}]{hunter_li_2005}
Hunter, D.~R. and R.~Li (2005).
\newblock Variable selection using {MM} algorithms.
\newblock {\em Ann. Statist.\/}~{\em 33\/}(4), 1617--1642.

\bibitem[\protect\citeauthoryear{Kab{\'a}n}{Kab{\'a}n}{2013}]{kaban_2013}
Kab{\'a}n, A. (2013).
\newblock Fractional norm regularization: Learning with very few relevant
  features.
\newblock {\em IEEE transactions on neural networks and learning
  systems\/}~{\em 24\/}(6), 953--963.

\bibitem[\protect\citeauthoryear{Kab{\'a}n and Durrant}{Kab{\'a}n and
  Durrant}{2008}]{kaban_durrant_2008}
Kab{\'a}n, A. and R.~J. Durrant (2008).
\newblock Learning with ${L}_{q<1}$ vs ${L}_1$-norm regularisation with
  exponentially many irrelevant features.
\newblock In {\em Joint European Conference on Machine Learning and Knowledge
  Discovery in Databases}, pp.\  580--596. Springer.

\bibitem[\protect\citeauthoryear{Luenberger and Ye}{Luenberger and
  Ye}{2008}]{luenberger_ye_2008}
Luenberger, D.~G. and Y.~Ye (2008).
\newblock {\em Linear and nonlinear programming\/} (Third ed.).
\newblock International Series in Operations Research \& Management Science,
  116. Springer, New York.

\bibitem[\protect\citeauthoryear{Majindar}{Majindar}{1963}]{majindar_1963}
Majindar, K.~N. (1963).
\newblock On a factorisation of positive definite matrices.
\newblock {\em Canad. Math. Bull.\/}~{\em 6}, 405--407.

\bibitem[\protect\citeauthoryear{Park and Casella}{Park and
  Casella}{2008}]{park_casella_2008}
Park, T. and G.~Casella (2008).
\newblock The {B}ayesian lasso.
\newblock {\em J. Amer. Statist. Assoc.\/}~{\em 103\/}(482), 681--686.

\bibitem[\protect\citeauthoryear{Schmidt, Fung, and Rosales}{Schmidt
  et~al.}{2007}]{schmidt_fung_rosales_2007}
Schmidt, M., G.~Fung, and R.~Rosales (2007).
\newblock Fast optimization methods for l1 regularization: A comparative study
  and two new approaches.
\newblock In {\em European Conference on Machine Learning}, pp.\  286--297.
  Springer.

\bibitem[\protect\citeauthoryear{Styan}{Styan}{1973}]{styan_1973}
Styan, G. P.~H. (1973).
\newblock Hadamard products and multivariate statistical analysis.
\newblock {\em Linear Algebra and Appl.\/}~{\em 6}, 217--240.

\bibitem[\protect\citeauthoryear{Tibshirani}{Tibshirani}{1996}]{tibshirani_1996}
Tibshirani, R. (1996).
\newblock Regression shrinkage and selection via the lasso.
\newblock {\em J. Roy. Statist. Soc. Ser. B\/}~{\em 58\/}(1), 267--288.

\bibitem[\protect\citeauthoryear{Tibshirani, Saunders, Rosset, Zhu, and
  Knight}{Tibshirani et~al.}{2005}]{tibshirani_saunders_rosset_zhu_knight_2005}
Tibshirani, R., M.~Saunders, S.~Rosset, J.~Zhu, and K.~Knight (2005).
\newblock Sparsity and smoothness via the fused lasso.
\newblock {\em J. R. Stat. Soc. Ser. B Stat. Methodol.\/}~{\em 67\/}(1),
  91--108.

\bibitem[\protect\citeauthoryear{Tibshirani}{Tibshirani}{2013}]{tibshirani_2013}
Tibshirani, R.~J. (2013).
\newblock The lasso problem and uniqueness.
\newblock {\em Electron. J. Stat.\/}~{\em 7}, 1456--1490.

\bibitem[\protect\citeauthoryear{Weisstein}{Weisstein}{2016}]{weisstein_2016_npd}
Weisstein, E.~W. (2016).
\newblock Normal product distribution. {From MathWorld--A Wolfram Web
  Resource}.
\newblock Visited on 10/25/16.

\bibitem[\protect\citeauthoryear{Yuan and Lin}{Yuan and
  Lin}{2006}]{yuan_lin_2006}
Yuan, M. and Y.~Lin (2006).
\newblock Model selection and estimation in regression with grouped variables.
\newblock {\em J. R. Stat. Soc. Ser. B Stat. Methodol.\/}~{\em 68\/}(1),
  49--67.

\bibitem[\protect\citeauthoryear{Zhou, Liu, and Fang}{Zhou
  et~al.}{2015}]{zhou_liu_fang_2015}
Zhou, Z., K.~Liu, and J.~Fang (2015).
\newblock Bayesian compressive sensing using normal product priors.
\newblock {\em IEEE Signal Processing Letters\/}~{\em 22\/}(5), 583--587.

\bibitem[\protect\citeauthoryear{Zou and Li}{Zou and Li}{2008}]{zou_li_2008}
Zou, H. and R.~Li (2008).
\newblock One-step sparse estimates in nonconcave penalized likelihood models.
\newblock {\em Ann. Statist.\/}~{\em 36\/}(4), 1509--1533.

\end{thebibliography}

\end{document}